\documentclass{article}

\usepackage{times}
\usepackage{graphicx} 
\usepackage{subfigure}

\usepackage{natbib}

\usepackage{algorithm}
\usepackage{algorithmic}

\usepackage{hyperref}

\usepackage[accepted]{icml2015}

\usepackage[utf8]{inputenc}
\usepackage[T1]{fontenc}
\usepackage{array}

\usepackage{color}
\icmltitlerunning{Inferring Graphs from Cascades: A Sparse Recovery Framework}

\usepackage{amsmath, amsfonts, amssymb, amsthm, bbm}
\usepackage{verbatim}
\newcommand{\reals}{\mathbb{R}}

\renewcommand{\O}{\mathcal{O}}
\DeclareMathOperator{\E}{\mathbb{E}}
\let\P\relax
\DeclareMathOperator{\P}{\mathbb{P}}

\newcommand{\inprod}[2]{#1 \cdot #2}
\newcommand{\defeq}{\equiv}
\DeclareMathOperator*{\argmax}{argmax}
\DeclareMathOperator*{\argmin}{argmin}

\newtheorem{theorem}{Theorem}
\newtheorem{lemma}{Lemma}
\newtheorem{corollary}{Corollary}

\newtheorem{proposition}{Proposition}
\newtheorem{definition}{Definition}

\begin{document}

\twocolumn[
\icmltitle{Inferring Graphs from Cascades: A Sparse Recovery Framework}

\icmlauthor{Jean Pouget-Abadie}{jeanpougetabadie@g.harvard.edu}
\icmladdress{Harvard University}
\icmlauthor{Thibaut Horel}{thorel@seas.harvard.edu}
\icmladdress{Harvard University}

\icmlkeywords{Sparse Recovery, Cascade Models, Graph Inference, Networks,
Diffusion Processes}

\vskip 0.3in
]

\begin{abstract}
In the Network Inference problem, one seeks to recover the edges of an unknown
graph from the observations of cascades propagating over this graph.  In this
paper, we approach this problem from the sparse recovery perspective.  We
introduce a general model of cascades, including the voter model and the
independent cascade model, for which we provide the first algorithm which
recovers the graph's edges with high probability and ${\cal O}(s\log m)$
measurements where $s$ is the maximum degree of the graph and $m$ is the number
of nodes.  Furthermore, we show that our algorithm also recovers the edge
weights (the parameters of the diffusion process) and is robust in the context
of approximate sparsity. Finally we prove an almost matching lower bound of
$\Omega(s\log\frac{m}{s})$ and validate our approach empirically on synthetic
graphs.

\end{abstract}

\section{Introduction}

Graphs have been extensively studied for their propagative abilities:
connectivity, routing, gossip algorithms, etc.
A diffusion process taking place over a graph provides valuable information
about the presence and weights of its edges. \emph{Influence cascades} are a
specific type of diffusion processes in which a particular infectious behavior
spreads over the nodes of the graph.  By only observing the ``infection times''
of the nodes in the graph, one might hope to recover the underlying graph and
the parameters of the cascade model. This problem is known in the literature as
the \emph{Network Inference problem}.

More precisely, solving the Network Inference problem involves designing an
algorithm taking as input a set of observed cascades (realisations of the
diffusion process) and recovers with high probability a large fraction of the
graph's edges. The goal is then to understand the relationship between the
number of observations, the probability of success, and the accuracy of the
reconstruction.

The Network Inference problem can be decomposed and analyzed ``node-by-node''.
Thus, we will focus on a single node of degree $s$ and discuss how to identify
its parents among the $m$ nodes of the graph. Prior work has shown that the
required number of observed cascades is $\O(poly(s)\log m)$
\cite{Netrapalli:2012, Abrahao:13}.

A more recent line of research~\cite{Daneshmand:2014} has focused on applying
advances in sparse recovery to the network inference problem. Indeed, the graph
can be interpreted as a ``sparse signal'' measured through influence cascades
and then recovered.  The challenge is that influence cascade models typically
lead to non-linear inverse problems and the measurements (the state of the
nodes at different time steps) are usually correlated. The sparse recovery
literature suggests that $\Omega(s\log\frac{m}{s})$ cascade observations should
be sufficient to recover the graph~\cite{donoho2006compressed, candes2006near}.
However, the best known upper bound to this day is $\O(s^2\log
m)$~\cite{Netrapalli:2012, Daneshmand:2014}

The contributions of this paper are the following:
\vspace{-1em}
\begin{itemize}
    \item we formulate the Graph Inference problem in the context of
        discrete-time influence cascades as a sparse recovery problem for
        a specific type of Generalized Linear Model. This formulation notably
        encompasses the well-studied Independent Cascade Model and Voter Model.
        \vspace{-0.5em}
    \item we give an algorithm which recovers the graph's edges using $\O(s\log
        m)$ cascades. Furthermore, we show that our algorithm is also able to
        efficiently recover the edge weights (the parameters of the influence
        model) up to an additive error term,
        \vspace{-0.5em}
    \item we show that our algorithm is robust in cases where the signal to
        recover is approximately $s$-sparse by proving guarantees in the
        \emph{stable recovery} setting.
        \vspace{-0.5em}
    \item we provide an almost tight lower bound of $\Omega(s\log \frac{m}{s})$
        observations required for sparse recovery.
\end{itemize}
\vspace{-0.5em}

The organization of the paper is as follows: we conclude the introduction by a
survey of the related work. In Section~\ref{sec:model} we present our model of
Generalized Linear Cascades and the associated sparse recovery formulation.  Its
theoretical guarantees are presented for various recovery settings in
Section~\ref{sec:results}. The lower bound is presented in
Section~\ref{sec:lowerbound}. Finally, we conclude with experiments in
Section~\ref{sec:experiments}.

\paragraph{Related Work}

The study of edge prediction in graphs has been an active field of research for
over a decade~\cite{Nowell08, Leskovec07, AdarA05}.~\cite{GomezRodriguez:2010}
introduced the {\scshape Netinf} algorithm, which approximates the likelihood of
cascades represented as a continuous process.  The algorithm was improved in
later work~\cite{gomezbalduzzi:2011}, but is not known to have any theoretical
guarantees beside empirical validation on synthetic networks.
\citet{Netrapalli:2012} studied the discrete-time version of the independent
cascade model and obtained the first ${\cal O}(s^2 \log m)$ recovery guarantee
on general networks. The algorithm is based on a likelihood function similar to
the one we propose, without the $\ell_1$-norm penalty. Their analysis depends on
a {\it correlation decay\/} assumption, which limits the number of new infections
at every step. In this setting, they show a lower bound of the number of
cascades needed for support recovery with constant probability of the order
$\Omega(s \log (m/s))$. They also suggest a {\scshape Greedy} algorithm, which
achieves a ${\cal O}(s \log m)$ guarantee in the case of tree graphs. The work
of~\cite{Abrahao:13} studies the same continuous-model framework as
\cite{GomezRodriguez:2010} and obtains an ${\cal O}(s^9 \log^2 s \log m)$
support recovery algorithm, without the \emph{correlation decay} assumption.
\cite{du2013uncover} propose a similar algorithm to ours for recovering the
weights of the graph under a continuous-time independent cascade model, without
proving theoretical guarantees.

Closest to this work is a recent paper by \citet{Daneshmand:2014}, wherein the
authors consider a $\ell_1$-regularized objective function. They adapt standard
results from sparse recovery to obtain a recovery bound of ${\cal O}(s^3 \log
m)$ under an irrepresentability condition~\cite{Zhao:2006}. Under stronger
assumptions, they match the~\cite{Netrapalli:2012} bound of ${\cal O}(s^2 \log
m)$, by exploiting similar properties of the convex program's KKT conditions.
In contrast, our work studies discrete-time diffusion processes including the
Independent Cascade model under weaker assumptions. Furthermore, we analyze both
the recovery of the graph's edges and the estimation of the model's parameters,
and achieve close to optimal bounds.

The work of~\cite{du2014influence} is slightly orthogonal to ours since they
suggest learning the \emph{influence} function, rather than the
parameters of the network directly.

\section{Model}
\label{sec:model}
We consider a graph ${\cal G}= (V, E, \Theta)$, where $\Theta$ is a $|V|\times
|V|$ matrix of parameters describing the edge weights of $\mathcal{G}$.
Intuitively, $\Theta_{i,j}$ captures the ``influence'' of node $i$ on node $j$.
Let $m\defeq |V|$. For each node $j$, let $\theta_{j}$ be the $j^{th}$ column
vector of $\Theta$.  A discrete-time \emph{Cascade model} is a
Markov process over a finite state space ${\{0, 1, \dots, K-1\}}^V$ with the
following properties:
\begin{enumerate}
    \item Conditioned on the previous time step, the transition events between
        two states in $\{0,1,\dots, K-1\}$ for each $i \in V$ are mutually
        independent across $i\in V$.
\item Of the $K$ possible states, there exists a \emph{contagious state} such
    that all transition probabilities of the Markov process can be expressed as
    a function of the graph parameters $\Theta$ and the set of ``contagious
    nodes'' at the previous time step.
  \item The initial probability over ${\{0, 1, \dots, K-1\}}^V$ is such that all
  nodes can eventually reach a \emph{contagious state} with non-zero
  probability. The ``contagious'' nodes at $t=0$ are called
    \emph{source nodes}.
\end{enumerate}

In other words, a cascade model describes a diffusion process where a set of
contagious nodes ``influence'' other nodes in the graph to become contagious.
An \emph{influence cascade} is a realisation of this random process,
\emph{i.e.} the successive states of the nodes in graph ${\cal G}$. Note that
both the ``single source'' assumption made in~\cite{Daneshmand:2014} and
\cite{Abrahao:13} as well as the ``uniformly chosen source set'' assumption
made in~\cite{Netrapalli:2012} verify condition 3. Also note that the
multiple-source node assumption does not reduce to the single-source
assumption, even under the assumption that cascades do not overlap. Imagining
for example two cascades starting from two different nodes; since we do not
observe which node propagated the contagion to which node, we cannot attribute
an infected node to either cascade and treat the problem as two independent
cascades.

In the context of Network Inference,~\cite{Netrapalli:2012} focus
on the well-known discrete-time independent cascade model recalled below, which
\cite{Abrahao:13} and~\cite{Daneshmand:2014} generalize to continuous time. We
extend the independent cascade model in a different direction by considering
a more general class of transition probabilities while staying in the
discrete-time setting. We observe that despite their obvious differences, both
the independent cascade and the voter models make the network inference problem
similar to the standard generalized linear model inference problem. In fact, we
define a class of diffusion processes for which this is true: the
\emph{Generalized Linear Cascade Models}. The linear threshold model is
a special case and is discussed in Section~\ref{sec:linear_threshold}.

\subsection{Generalized Linear Cascade Models}
\label{sec:GLC}

Let \emph{susceptible} denote any state which can become contagious at
the next time step with a non-zero probability. We draw inspiration from
generalized linear models to introduce Generalized Linear Cascades:

\begin{definition}
\label{def:glcm}
Let $X^t$ be the indicator variable of ``contagious nodes'' at time step $t$.
A \emph{generalized linear cascade model} is a cascade model such that for each
susceptible node $j$ in state $s$ at time step $t$, the probability of $j$
becoming ``contagious'' at time step $t+1$ conditioned on $X^t$ is a Bernoulli
variable of parameter $f(\theta_j \cdot X^t)$:
\begin{equation}
    \label{eq:glm}
    \mathbb{P}(X^{t+1}_j = 1|X^t)
    = f(\theta_j \cdot X^t)
\end{equation}
where $f: \reals \rightarrow [0,1]$
\end{definition}

In other words, each generalized linear cascade provides, for each node $j \in
V$ a series of measurements ${(X^t, X^{t+1}_j)}_{t \in {\cal T}_j}$ sampled from
a generalized linear model. Note also that $\E[X^{t+1}_i\,|\,X^t]
= f(\inprod{\theta_i}{X^t})$. As such, $f$ can be interpreted as the inverse
link function of our generalized linear cascade model.




\subsection{Examples}

\subsubsection{Independent Cascade Model}

In the independent cascade model, nodes can be either susceptible, contagious
or immune. At $t=0$, all source nodes are ``contagious'' and all remaining
nodes are ``susceptible''. At each time step $t$, for each edge $(i,j)$ where
$j$ is susceptible and $i$ is contagious, $i$ attempts to infect $j$ with
probability $p_{i,j}\in[0,1]$; the infection attempts are mutually independent.
If $i$ succeeds, $j$ will become contagious at time step $t+1$. Regardless of
$i$'s success, node $i$ will be immune at time $t+1$, such that nodes
stay contagious for only one time step. The cascade process terminates when no
contagious nodes remain.

If we denote by $X^t$ the indicator variable of the set of contagious nodes at
time step $t$, then if $j$ is susceptible at time step $t+1$, we have:
\begin{displaymath}
    \P\big[X^{t+1}_j = 1\,|\, X^{t}\big]
    = 1 - \prod_{i = 1}^m {(1 - p_{i,j})}^{X^t_i}.
\end{displaymath}
Defining $\Theta_{i,j} \defeq \log(\frac{1}{1-p_{i,j}})$, this can be rewritten as:
\begin{align*}\label{eq:ic}
    \tag{IC}
    \P\big[X^{t+1}_j = 1\,|\, X^{t}\big]
    &= 1 - \prod_{i = 1}^m e^{-\Theta_{i,j}X^t_i}\\
    &= 1 - e^{-\inprod{\Theta_j}{X^t}}
\end{align*}

Therefore, the independent cascade model is a Generalized Linear Cascade model
with inverse link function $f : z \mapsto 1 - e^{-z}$. Note that to write the
Independent Cascade Model as a Generalized Linear Cascade Model, we had to
introduce the change of variable $\Theta_{i,j} = \log(\frac{1}{1-p_{i,j}})$.
The recovery results in Section~\ref{sec:results} pertain to the $\Theta_j$
parameters. Fortunately, the following lemma shows that the recovery error on
$\Theta_j$ is an upper bound on the error on the original $p_j$ parameters.

\begin{lemma}
    \label{lem:transform}
    $\|\hat{\theta} - \theta^* \|_2 \geq \|\hat{p} - p^*\|_2$.
\end{lemma}

\subsubsection{The Linear Voter Model}

In the Linear Voter Model, nodes can be either \emph{red} or \emph{blue}.
Without loss of generality, we can suppose that the \emph{blue} nodes are
contagious. The parameters of the graph are normalized such that $\forall i, \
\sum_j \Theta_{i,j} = 1$.  Each round, every node $j$
independently chooses one of its neighbors with probability $\Theta_{i,j}$ and
adopts their color. The cascades stops at a fixed horizon time $T$ or if all
nodes are of the same color.  If we denote by $X^t$ the indicator variable of
the set of blue nodes at time step $t$, then we have:
\begin{equation}
\mathbb{P}\left[X^{t+1}_j = 1 | X^t \right] = \sum_{i=1}^m \Theta_{i,j} X_i^t =
\inprod{\Theta_j}{X^t}
\tag{V}
\end{equation}

Thus, the linear voter model is a Generalized Linear Cascade model
with inverse link function $f: z \mapsto z$.

\subsubsection{Discretization of Continuous Model}

Another motivation for the Generalized Linear Cascade model is that it captures
the time-discretized formulation of the well-studied  continuous-time
independent cascade model with exponential transmission function (CICE)
of~\cite{GomezRodriguez:2010, Abrahao:13, Daneshmand:2014}. Assume that the
temporal resolution of the discretization is $\varepsilon$, \emph{i.e.} all
nodes whose (continuous) infection time is within the interval $[k\varepsilon,
    (k+1)\varepsilon)$ are considered infected at (discrete) time step $k$. Let
    $X^k$ be the indicator vector of the set of nodes `infected' before or
    during the $k^{th}$ time interval.  Note that contrary to the discrete-time
    independent cascade model, $X^k_j = 1 \implies X^{k+1}_j = 1$, that is,
    there is no immune state and nodes remain contagious forever.

Let $\text{Exp}(p)$ be an exponentially-distributed random variable of parameter $p$
and let $\Theta_{i,j}$ be the rate of transmission along directed edge $(i,j)$
in the CICE model.  By the memoryless property of the exponential, if $X^k_j
\neq 1$:
\begin{multline*}
  \mathbb{P}(X^{k+1}_j = 1 | X^k) = \mathbb{P}(\min_{i \in {\cal N}(j)}
    \text{Exp}(\Theta_{i,j}) \leq \epsilon) \\
    = \mathbb{P}(\text{Exp}( \sum_{i=1}^m \Theta_{i,j} X^t_i) \leq \epsilon)
  = 1 - e^{- \epsilon \Theta_j \cdot X^t}
\end{multline*}
Therefore, the $\epsilon$-discretized CICE-induced process is a
Generalized Linear Cascade model with inverse link function $f:z\mapsto
1-e^{-\epsilon\cdot z}$.

\subsubsection{Logistic Cascades}
\label{sec:logistic_cascades}
``Logistic cascades'' is the specific case where the inverse link function is
given by the logistic function
$f(z) = 1/(1+e^{-z + t})$.
Intuitively, this captures the idea that there is a threshold $t$ such that
when the sum of the parameters of the infected parents of a node is larger than
the threshold, the probability of getting infected is close to one. This is
a smooth approximation of the hard threshold rule of the Linear Threshold
Model~\cite{Kempe:03}. As we will see later in the analysis, for logistic
cascades, the graph inference problem becomes a linear inverse problem.





\subsection{Maximum Likelihood Estimation}
\label{sec:mle}

\begin{figure}
  \includegraphics[scale=.4]{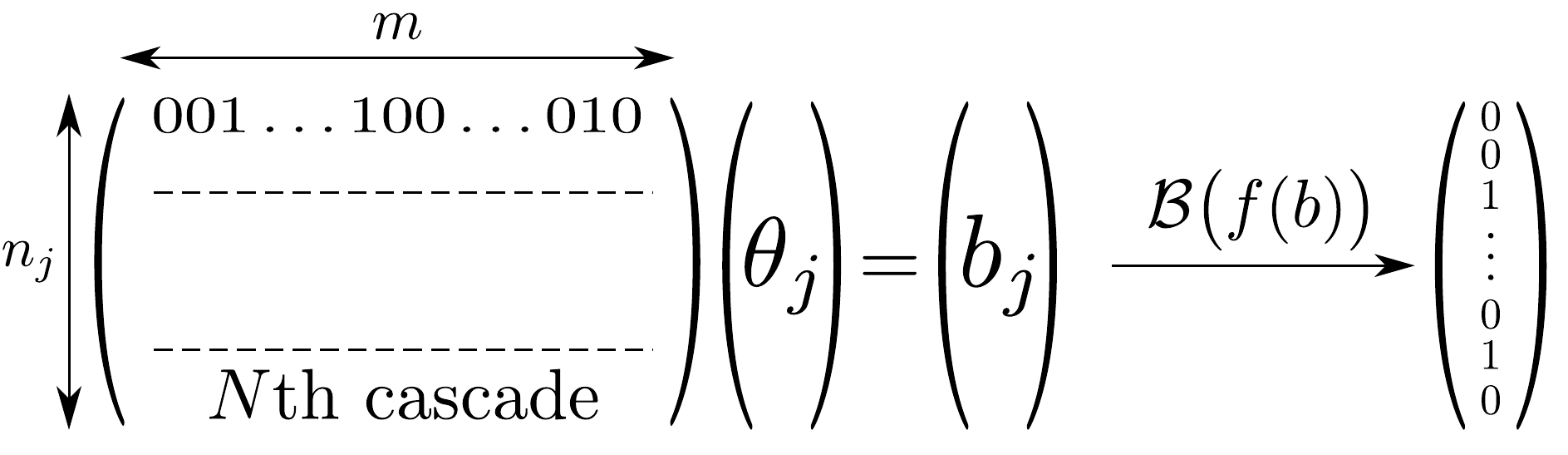}
  \caption{Illustration of the sparse-recovery approach. Our objective is to
  recover the unknown weight vector $\theta_j$ for each node $j$. We observe a
Bernoulli realization whose parameters are given by applying $f$ to the
matrix-vector product, where the measurement matrix encodes which nodes are
``contagious'' at each time step.}
\vspace{-1em}
\end{figure}

Inferring the model parameter $\Theta$ from observed influence cascades is the
central question of the present work. Recovering the edges in $E$ from observed
influence cascades is a well-identified problem known as the \emph{Network
Inference} problem. However, recovering the influence parameters is no less
important. In this work we focus on recovering $\Theta$, noting that the set of
edges $E$ can then be recovered through the following equivalence: $(i,j)\in
E\Leftrightarrow  \Theta_{i,j} \neq 0$

Given observations $(x^1,\ldots,x^n)$ of a cascade model, we can recover
$\Theta$ via Maximum Likelihood Estimation (MLE). Denoting by $\mathcal{L}$ the
log-likelihood function, we consider the following $\ell_1$-regularized MLE
problem:
\begin{displaymath}
    \hat{\Theta} \in \argmax_{\Theta} \frac{1}{n}
    \mathcal{L}(\Theta\,|\,x^1,\ldots,x^n) - \lambda\|\Theta\|_1
\end{displaymath}
where $\lambda$ is the regularization factor which helps prevent
overfitting and controls the sparsity of the solution.

The generalized linear cascade model is decomposable in the following sense:
given Definition~\ref{def:glcm}, the log-likelihood can be written as the sum
of $m$ terms, each term $i\in\{1,\ldots,m\}$ only depending on $\theta_i$.
Since this is equally true for $\|\Theta\|_1$, each column $\theta_i$ of
$\Theta$ can be estimated by a separate optimization program:
\begin{equation}\label{eq:pre-mle}
    \hat{\theta}_i \in \argmax_{\theta} \mathcal{L}_i(\theta_i\,|\,x^1,\ldots,x^n)
    - \lambda\|\theta_i\|_1
\end{equation}
where we denote by ${\cal T}_i$ the time steps at which node $i$ is susceptible
and:
\begin{multline*}
    \mathcal{L}_i(\theta_i\,|\,x^1,\ldots,x^n) = \frac{1}{|{\cal T}_i|}
    \sum_{t\in {\cal T}_i } x_i^{t+1}\log f(\inprod{\theta_i}{x^{t}}) \\ + (1 -
    x_i^{t+1})\log\big(1-f(\inprod{\theta_i}{x^t})\big)
\end{multline*}

In the case of the voter model, the measurements include all time steps until
we reach the time horizon $T$ or the graph coalesces to a single state. For the
independent cascade model, the measurements include all time steps until node
$i$ becomes contagious, after which its behavior is deterministic.  Contrary to
prior work, our results depend on the number of measurements and not the number
of cascades.

\paragraph{Regularity assumptions}

To solve program~\eqref{eq:pre-mle} efficiently, we would like it to be convex.
A sufficient condition is to assume that $\mathcal{L}_i$ is concave, which is
the case if $f$ and $(1-f)$ are both log-concave. Remember that a
twice-differentiable function $f$ is log-concave iff. $f''f \leq f'^2$.  It is
easy to verify this property for $f$ and $(1-f)$ in the Independent Cascade
Model and Voter Model.

Furthermore, the data-dependent bounds in Section~\ref{sec:main_theorem} will
require the following regularity assumption on the inverse link function $f$:
there exists $\alpha\in(0,1)$ such that
\begin{equation}
  \tag{LF}
  \max \big\{ | (\log f)'(z_x) |, |(\log (1-f))'(z_x) | \big\}
  \leq \frac{1}{\alpha}
\end{equation}
for all $z_x\defeq\inprod{\theta^*}{x}$ such that
$f(z_x)\notin\{0,1\}$.

In the voter model, $\frac{f'(z)}{f(z)} = \frac{1}{z}$ and
$\frac{f'(z)}{(1-f)(z)} =
\frac{1}{1-z}$. Hence (LF) will hold as soon as $\alpha\leq \Theta_{i,j}\leq
1-\alpha$ for all $(i,j)\in E$ which is always satisfied for some $\alpha$ for
non-isolated nodes.  In the Independent Cascade Model, $\frac{f'(z)}{f(z)} =
\frac{1}{e^{z}-1}$ and $\frac{f'(z)}{(1-f)(z)} = 1$. Hence (LF) holds as soon
as $p_{i,j}\geq \alpha$ for all $(i,j)\in E$ which is always satisfied for some
$\alpha\in(0,1)$.

For the data-independent bound of Proposition~\ref{prop:fi}, we will require the
following additional regularity assumption:
\begin{equation}
  \tag{LF2}
  \max \big\{ | (\log f)''(z_x) |, |(\log (1-f))''(z_x) | \big\}
  \leq \frac{1}{\alpha}
\end{equation}
for some $\alpha\in(0, 1)$ and for all $z_x\defeq\inprod{\theta^*}{x}$ such
that $f(z_x)\notin\{0,1\}$. It is again easy to see that this condition is
verified for the Independent Cascade Model and the Voter model for the same
$\alpha\in(0,1)$.

\paragraph{Convex constraints} The voter model is only defined when
$\Theta_{i,j}\in (0,1)$ for all $(i,j)\in E$. Similarly the independent cascade
model is only defined when $\Theta_{i,j}> 0$. Because the likelihood function
$\mathcal{L}_i$ is equal to $-\infty$ when the parameters are outside of the
domain of definition of the models, these contraints do not need to appear
explicitly in the optimization program.

In the specific case of the voter model, the constraint $\sum_j \Theta_{i,j}
= 1$ will not necessarily be verified by the estimator obtained in
\eqref{eq:pre-mle}. In some applications, the experimenter might not need this
constraint to be verified, in which case the results in
Section~\ref{sec:results} still give a bound on the recovery error. If this
constraint needs to be satisfied, then by Lagrangian duality, there exists
a $\lambda\in \reals$ such that adding $\lambda\big(\sum_{j}\theta_j
- 1\big)$ to the objective function of~\eqref{eq:pre-mle} enforces the
constraint. Then, it suffices to apply the results of Section~\ref{sec:results}
to the augmented objective to obtain the same recovery guarantees. Note that
the added term is linear and will easily satisfy all the required regularity
assumptions.

\section{Results}
\label{sec:results}
In this section, we apply the sparse recovery framework to analyze under which
assumptions our program~\eqref{eq:pre-mle} recovers the true parameter
$\theta_i$ of the cascade model. Furthermore, if we can estimate $\theta_i$ to
a sufficiently good accuracy, it is then possible to recover the support of
$\theta_i$ by simple thresholding, which provides a solution to the standard
Network Inference problem.

We will first give results in the exactly sparse setting in which $\theta_i$
has a support of size exactly $s$. We will then relax this sparsity constraint
and give results in the \emph{stable recovery} setting where $\theta_i$ is
approximately $s$-sparse.

As mentioned in Section~\ref{sec:mle}, the maximum likelihood estimation
program is decomposable. We will henceforth focus on a single node $i\in V$ and
omit the subscript $i$ in the notations when there is no ambiguity. The
recovery problem is now the one of estimating a single vector $\theta^*$ from
a set $\mathcal{T}$ of observations. We will write $n\defeq |\mathcal{T}|$.

\subsection{Main Theorem}
\label{sec:main_theorem}

In this section, we analyze the case where $\theta^*$ is exactly sparse. We
write $S\defeq\text{supp}(\theta^*)$ and $s=|S|$. Recall, that $\theta_i$ is the
vector of weights for all edges \emph{directed at} the node we are solving for.
In other words, $S$ is the set of all nodes susceptible to influence node $i$,
also referred to as its parents. Our main theorem will rely on the now standard
\emph{restricted eigenvalue condition} introduced
by~\cite{bickel2009simultaneous}.

\begin{definition}
    Let $\Sigma\in\mathcal{S}_m(\reals)$ be a real symmetric matrix and $S$ be
    a subset of $\{1,\ldots,m\}$. Defining $\mathcal{C}(S)\defeq
    \{X\in\reals^m\,:\,\|X_{S^c}\|_1\leq
    3\|X_S\|_1\}$. We say that $\Sigma$ satisfies the
    $(S,\gamma)$-\emph{restricted eigenvalue condition} iff:
\begin{equation}
    \forall X \in {\cal C(S)}, X^T \Sigma X \geq \gamma \|X\|_2^2
\tag{RE}
\label{eq:re}
\end{equation}
\end{definition}

A discussion of the $(S,\gamma)$-{\bf(RE)} assumption in the context of
generalized linear cascade models can be found in Section~\ref{sec:re}. In our
setting we require that the {\bf(RE)}-condition holds for the Hessian of the
log-likelihood function $\mathcal{L}$: it essentially captures the fact that
the binary vectors of the set of active nodes (\emph{i.e} the measurements) are
not \emph{too} collinear.

\begin{theorem}
\label{thm:main}
Assume the Hessian $\nabla^2\mathcal{L}(\theta^*)$ satisfies the
$(S,\gamma)$-{\bf(RE)} for some $\gamma > 0$ and that {\bf (LF)} holds for some
$\alpha > 0$. For any $\delta\in(0,1)$, let $\hat{\theta}$ be the solution of
\eqref{eq:pre-mle} with $\lambda \defeq 2\sqrt{\frac{\log m}{\alpha n^{1
- \delta}}}$, then:
\begin{equation}
    \label{eq:sparse}
    \|\hat \theta - \theta^* \|_2
    \leq \frac{6}{\gamma} \sqrt{\frac{s \log m}{\alpha n^{1-\delta}}}
\quad
\text{w.p.}\;1-\frac{1}{e^{n^\delta \log m}}
\end{equation}
\end{theorem}

Note that we have expressed the convergence rate in the number of measurements
$n$, which is different from the number of cascades. For example, in the case
of the voter model with horizon time $T$ and for $N$ cascades, we can expect
a number of measurements proportional to $N\times T$.

Theorem~\ref{thm:main} is a consequence of Theorem~1 in \cite{Negahban:2009}
which gives a bound on the convergence rate of regularized estimators. We state
their theorem in the context of $\ell_1$ regularization in
Lemma~\ref{lem:negahban}.

\begin{lemma} \label{lem:negahban}
Let ${\cal C}(S) \defeq \{ \Delta \in \mathbb{R}^m\,|\,\|\Delta_S\|_1 \leq
3 \|\Delta_{S^c}\|_1 \}$. Suppose that:
\begin{multline}
    \label{eq:rc}
    \forall \Delta \in {\cal C}(S), \;
    {\cal L}(\theta^* + \Delta) - {\cal L}(\theta^*)\\
    - \inprod{\nabla {\cal L}(\theta^*)}{\Delta}
    \geq \kappa_{\cal L} \|\Delta\|_2^2 - \tau_{\cal L}^2(\theta^*)
\end{multline}
for some $\kappa_{\cal L} > 0$ and function $\tau_{\cal L}$. Finally suppose
that $\lambda \geq 2 \|\nabla {\cal L}(\theta^*)\|_{\infty}$, then if
$\hat{\theta}_\lambda$ is the solution of \eqref{eq:pre-mle}:
\begin{displaymath}
\|\hat \theta_\lambda - \theta^* \|_2^2
\leq 9 \frac{\lambda^2 s}{\kappa_{\cal L}}
+ \frac{\lambda}{\kappa_{\cal L}^2} 2 \tau^2_{\cal L}(\theta^*)
\end{displaymath}
\end{lemma}

To prove Theorem~\ref{thm:main}, we apply Lemma~\ref{lem:negahban} with
$\tau_{\mathcal{L}}=0$. Since $\mathcal{L}$ is twice differentiable and convex,
assumption \eqref{eq:rc} with $\kappa_{\mathcal{L}}=\frac{\gamma}{2}$ is implied
by the (RE)-condition. For a good convergence rate, we must find the smallest
possible value of $\lambda$ such that $\lambda \geq 2
\|\nabla\mathcal{L}\theta^*\|_{\infty}$.  The upper bound on the $\ell_{\infty}$
norm of $\nabla\mathcal{L}(\theta^*)$ is given by Lemma~\ref{lem:ub}.

\begin{lemma}
\label{lem:ub}
Assume {\bf(LF)} holds for some $\alpha>0$. For any $\delta\in(0,1)$:
\begin{displaymath}
    \|\nabla {\cal L}(\theta^*)\|_{\infty}
    \leq 2 \sqrt{\frac{\log m}{\alpha n^{1 - \delta}}}
    \quad
    \text{w.p.}\; 1-\frac{1}{e^{n^\delta \log m}}
\end{displaymath}
\end{lemma}

The proof of Lemma~\ref{lem:ub} relies crucially on Azuma-Hoeffding's
inequality, which allows us to handle correlated observations. This departs
from the usual assumptions made in sparse recovery settings, that the
measurements are independent from one another. We now show how to
use Theorem~\ref{thm:main} to recover the support of $\theta^*$, that is, to
solve the Network Inference problem.

\begin{corollary}
\label{cor:variable_selection}
Under the same assumptions as Theorem~\ref{thm:main}, let $\hat {\cal S}_\eta
\defeq \{ j \in \{1,\ldots, m\} : \hat{\theta}_j > \eta\}$ for $\eta > 0$. For
$0< \epsilon < \eta$, let ${\cal S}^*_{\eta + \epsilon} \defeq \{ i \in
\{1,\ldots,m\} :\theta_i^* > \eta +\epsilon \}$ be the set of all true `strong'
parents. Suppose the number of measurements verifies: $ n > \frac{9s\log
m}{\alpha\gamma^2\epsilon^2}$.  Then with probability $1-\frac{1}{m}$, ${\cal
S}^*_{\eta + \epsilon} \subseteq \hat {\cal S}_\eta \subseteq {\cal S}^*$. In
other words we recover all `strong' parents and no `false' parents.
\end{corollary}

Assuming we know a lower bound $\alpha$ on $\Theta_{i,j}$,
Corollary~\ref{cor:variable_selection} can be applied to the Network Inference
problem in the following manner: pick $\epsilon = \frac{\eta}{2}$ and $\eta
= \frac{\alpha}{3}$, then $S_{\eta+\epsilon}^* = S$ provided that
$n=\Omega\left(\frac{s\log m}{\alpha^3\gamma^2}\right)$. That is, the support
of $\theta^*$ can be found by thresholding $\hat{\theta}$ to the level $\eta$.

\subsection{Approximate Sparsity}
\label{sec:relaxing_sparsity}

In practice, exact sparsity is rarely verified. For social networks in
particular, it is more realistic to assume that each node has few ``strong''
parents' and many ``weak'' parents. In other words, even if $\theta^*$ is not
exactly $s$-sparse, it can be well approximated by $s$-sparse vectors.

Rather than obtaining an impossibility result, we show that the bounds obtained
in Section~\ref{sec:main_theorem} degrade gracefully in this setting. Formally,
let
$
    \theta^*_{\lfloor s \rfloor}
    \in \argmin_{\|\theta\|_0 \leq s} \|\theta - \theta^*\|_1
    $
be the best $s$-approximation to $\theta^*$. Then we pay a cost proportional
to $\|\theta^* - \theta^*_{\lfloor s\rfloor}\|_1$ for recovering the weights of
non-exactly sparse vectors. This cost is simply the ``tail'' of
$\theta^*$: the sum of the $m-s$ smallest coordinates of $\theta^*$. We recover
the results of Section~\ref{sec:main_theorem} in the limit of exact sparsity.
These results are formalized in the following theorem, which is also a
consequence of Theorem 1 in \cite{Negahban:2009}.
\begin{theorem}
\label{thm:approx_sparse}
Suppose the {\bf(RE)} assumption holds for the Hessian $\nabla^2 f(\theta^*)$
and $\tau_{\mathcal{L}}(\theta^*) = \frac{\kappa_2\log m}{n}\|\theta^*\|_1$ on
the following set:
\begin{align}
\nonumber
{\cal C}' \defeq & \{X \in \mathbb{R}^p : \|X_{S^c}\|_1 \leq 3 \|X_S\|_1
+ 4 \|\theta^* - \theta^*_{\lfloor s \rfloor}\|_1 \} \\ \nonumber
& \cap \{ \|X\|_1 \leq 1 \}
\end{align}
If the number of measurements $n\geq \frac{64\kappa_2}{\gamma}s\log m$, then by
solving \eqref{eq:pre-mle} for $\lambda \defeq 2\sqrt{\frac{\log m}{\alpha n^{1
- \delta}}}$ we have:
\begin{align*}
    \|\hat \theta - \theta^* \|_2 \leq \frac{3}{\gamma} \sqrt{\frac{s\log
    m}{\alpha n^{1-\delta}}} + 4 \sqrt[4]{\frac{s\log m}{\gamma^4\alpha
    n^{1-\delta}}} \|\theta^* - \theta^*_{\lfloor s \rfloor}\|_1
\end{align*}
\end{theorem}

As in Corollary~\ref{cor:variable_selection}, an edge recovery guarantee can be
derived from  Theorem~\ref{thm:approx_sparse} in the case of approximate
sparsity.

\subsection{Restricted Eigenvalue Condition}
\label{sec:re}

There exists a large class of sufficient conditions under which sparse recovery
is achievable in the context of regularized estimation~\cite{vandegeer:2009}.
The restricted eigenvalue condition, introduced in \cite{bickel:2009}, is one
of the weakest such assumption. It can be interpreted as a restricted form of
non-degeneracy. Since we apply it to the Hessian of the log-likelihood function
$\nabla^2 \mathcal{L}(\theta)$, it essentially reduces to a form of restricted
strong convexity, that Lemma~\ref{lem:negahban} ultimately relies on.

Observe that the Hessian of $\mathcal{L}$ can be seen as a re-weighted
\emph{Gram matrix} of the observations:
\begin{multline*}
    \nabla^2\mathcal{L}(\theta^*)
    = \frac{1}{|\mathcal{T}|}\sum_{t\in\mathcal{T}}x^t(x^t)^T
    \bigg[x_i^{t+1}\frac{f''f-f'^2}{f^2}(\inprod{\theta^*}{x^t})\\
    -(1-x_i^{t+1})\frac{f''(1-f) + f'^2}{(1-f)^2}(\inprod{\theta^*}{x^t})\bigg]
\end{multline*}
If $f$ and $(1-f)$ are $c$-strictly log-convex for $c>0$,
then $ \min\left((\log f)'', (\log (1-f))'' \right) \geq c $. This implies that
the $(S, \gamma)$-({\bf RE}) condition in Theorem~\ref{thm:main} and
Theorem~\ref{thm:approx_sparse} reduces to a condition on the \emph{Gram
matrix} of the observations $X^T
X = \frac{1}{|\mathcal{T}|}\sum_{t\in\mathcal{T}}x^t(x^t)^T$ for $\gamma'
\defeq \gamma\cdot c$.

\paragraph{(RE) with high probability}

The Generalized Linear Cascade model yields a probability distribution over the
observed sets of infected nodes $(x^t)_{t\in\mathcal{T}}$. It is then natural
to ask whether the restricted eigenvalue condition is likely to occur under
this probabilistic model. Several recent papers show that large classes of
correlated designs obey the restricted eigenvalue property with high
probability \cite{raskutti:10, rudelson:13}.

The {\bf(RE)}-condition has the following concentration property: if it holds
for the expected Hessian matrix $\E[\nabla^2\mathcal{L}(\theta^*)]$, then it
holds for the finite sample Hessian matrix $\nabla^2\mathcal{L}(\theta^*)$ with
high probability.

Therefore, under an assumption which only involves the probabilistic model and
not the actual observations, we can obtain the same conclusion as in
Theorem~\ref{thm:main}:

\begin{proposition}
    \label{prop:fi}
    Suppose $\E[\nabla^2\mathcal{L}(\theta^*)]$ verifies the $(S,\gamma)$-{\bf
    (RE)} condition and assume {\bf (LF)} and {\bf (LF2)}. For $\delta> 0$, if
    $n^{1-\delta}\geq \frac{1}{28\gamma\alpha}s^2\log m $, then
    $\nabla^2\mathcal{L}(\theta^*)$ verifies the $(S,\frac{\gamma}{2})$-(RE)
    condition, w.p $\geq 1-e^{-n^\delta\log m}$.
\end{proposition}

Observe that the number of measurements required in Proposition~\ref{prop:fi}
is now quadratic in $s$. If we only keep the first measurement from each
cascade, which are independent, we can apply Theorem 1.8 from
\cite{rudelson:13}, lowering the number of required cascades to $s\log m \log^3(
s\log m)$.

If $f$ and $(1-f)$ are strictly log-convex, then the previous observations show
that the quantity $\E[\nabla^2\mathcal{L}(\theta^*)]$ in
Proposition~\ref{prop:fi} can be replaced by the expected \emph{Gram matrix}:
$A \equiv \mathbb{E}[X^T X]$. This matrix $A$ has a natural interpretation: the
entry $a_{i,j}$ is the probability that node $i$ and node $j$ are infected at
the same time during a cascade. In particular, the diagonal term $a_{i,i}$ is
simply the probability that node $i$ is infected during a cascade.

\section{A Lower Bound}
\label{sec:lowerbound}
In \cite{Netrapalli:2012}, the authors explicitate a lower bound of
$\Omega(s\log\frac{m}{s})$ on the number of cascades necessary to achieve good
support recovery with constant probability under a \emph{correlation decay}
assumption.  In this section, we will consider the stable sparse recovery
setting of Section~\ref{sec:relaxing_sparsity}.  Our goal is to obtain an
information-theoretic lower bound on the number of measurements necessary to
approximately recover the parameter $\theta^*$ of a cascade model from observed
cascades. Similar lower bounds were obtained for sparse \emph{linear} inverse
problems in \cite{pw11, pw12, bipw11}.

\begin{theorem}
    \label{thm:lb}
    Let us consider a cascade model of the form \eqref{eq:glm} and a recovery
    algorithm $\mathcal{A}$ which takes as input $n$ random cascade
    measurements and outputs $\hat{\theta}$ such that with probability
    $\delta>\frac{1}{2}$ (over the measurements):
    \begin{equation}
        \label{eq:lb}
        \|\hat{\theta}-\theta^*\|_2\leq
        C\min_{\|\theta\|_0\leq s}\|\theta-\theta^*\|_2
    \end{equation}
    where $\theta^*$ is the true parameter of the cascade model. Then $n
    = \Omega(s\log\frac{m}{s}/\log C)$.
\end{theorem}

This theorem should be contrasted with Theorem~\ref{thm:approx_sparse}: up to
an additive $s\log s$ factor, the number of measurements required by our
algorithm is tight.  The proof of Theorem~\ref{thm:lb} follows an approach
similar to \cite{pw12}.  We present a sketch of the proof in the Appendix and
refer the reader to their paper for more details.

\section{Experiments}
\label{sec:experiments}


\begin{figure*}[t]
\centering
\begin{tabular}{l l l}
    \hspace{-0.5em}\includegraphics[scale=.28]{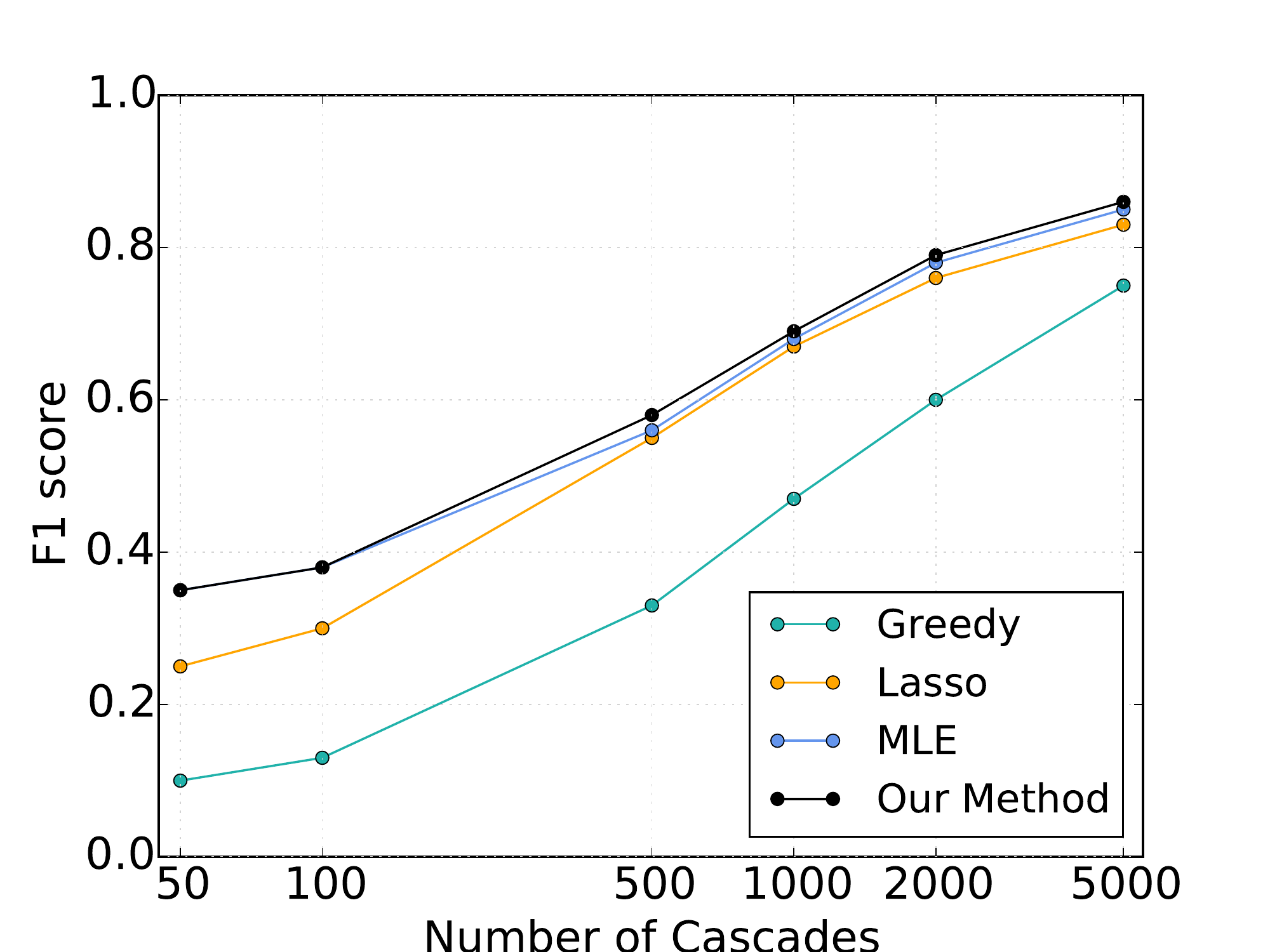}
& \hspace{-1em}\includegraphics[scale=.28]{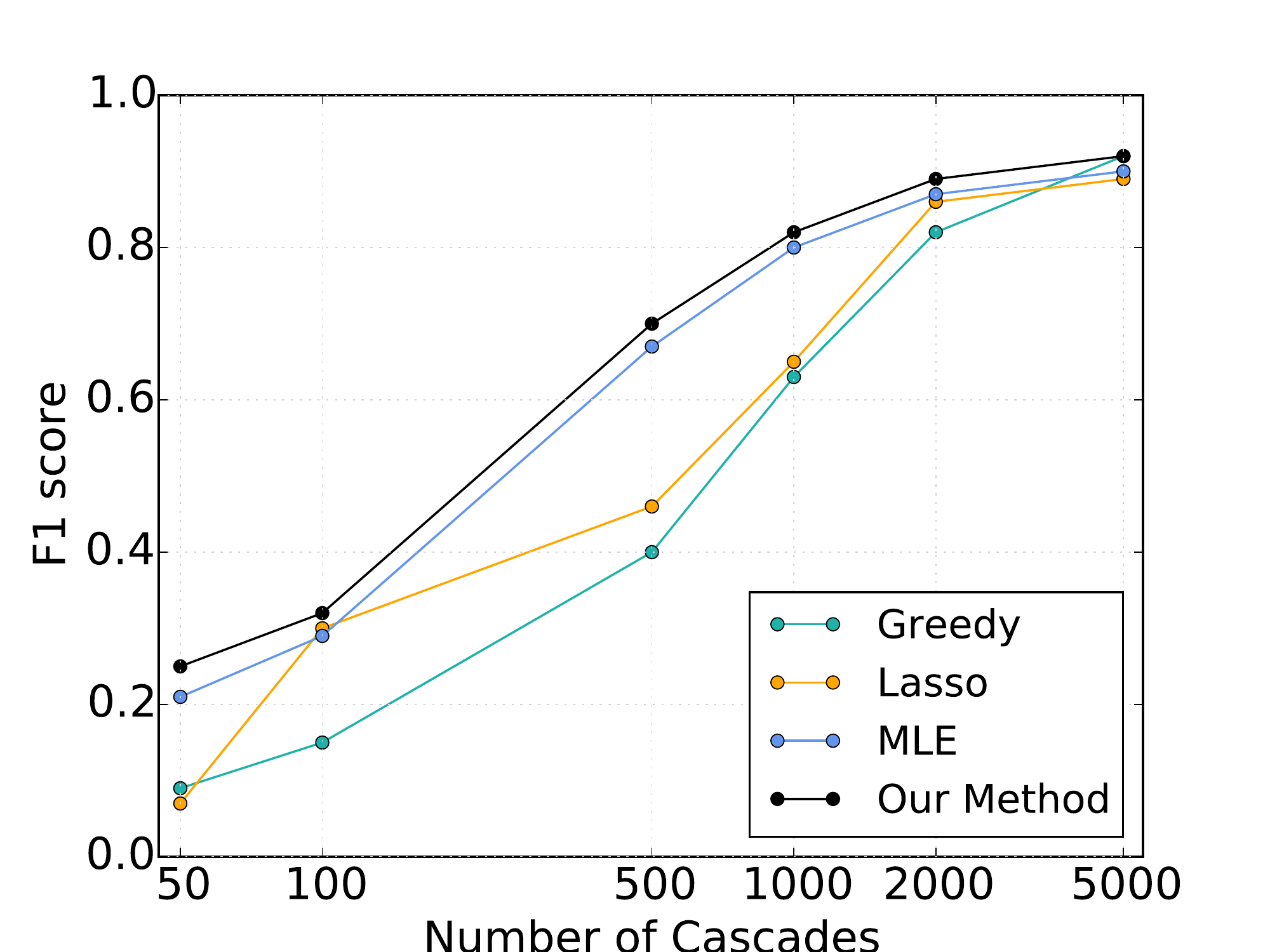}
& \hspace{-3em}\includegraphics[scale=.28]{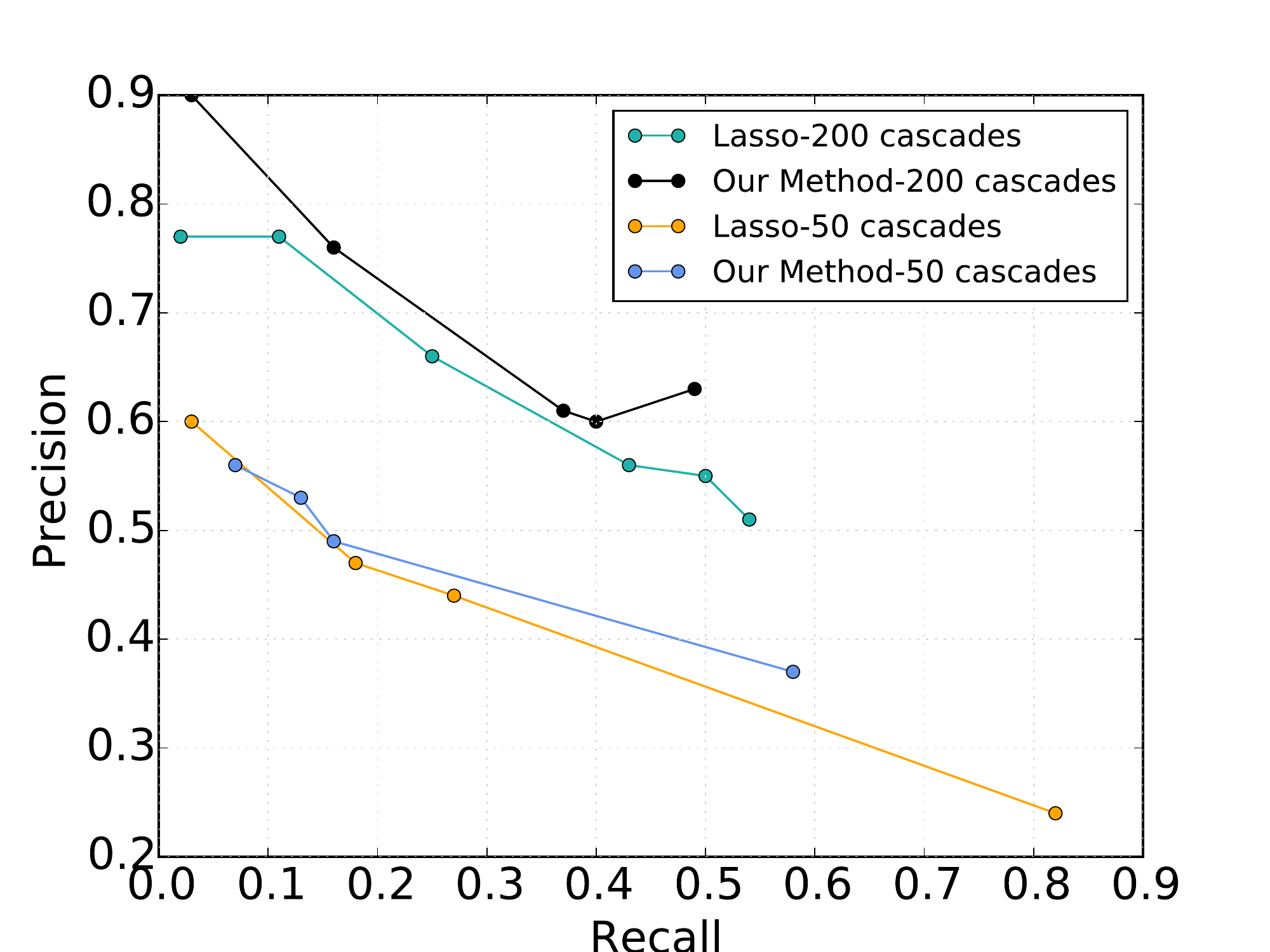}  \\

\hspace{-0.5em} (a) Barabasi-Albert (F$1$ \emph{vs.} $n$)
&\hspace{-1em} (b) Watts-Strogatz (F$1$ \emph{vs.} $n$)
&\hspace{-3em} (c) Holme-Kim (Prec-Recall) \\
 \hspace{-0.5em}\includegraphics[scale=.28]{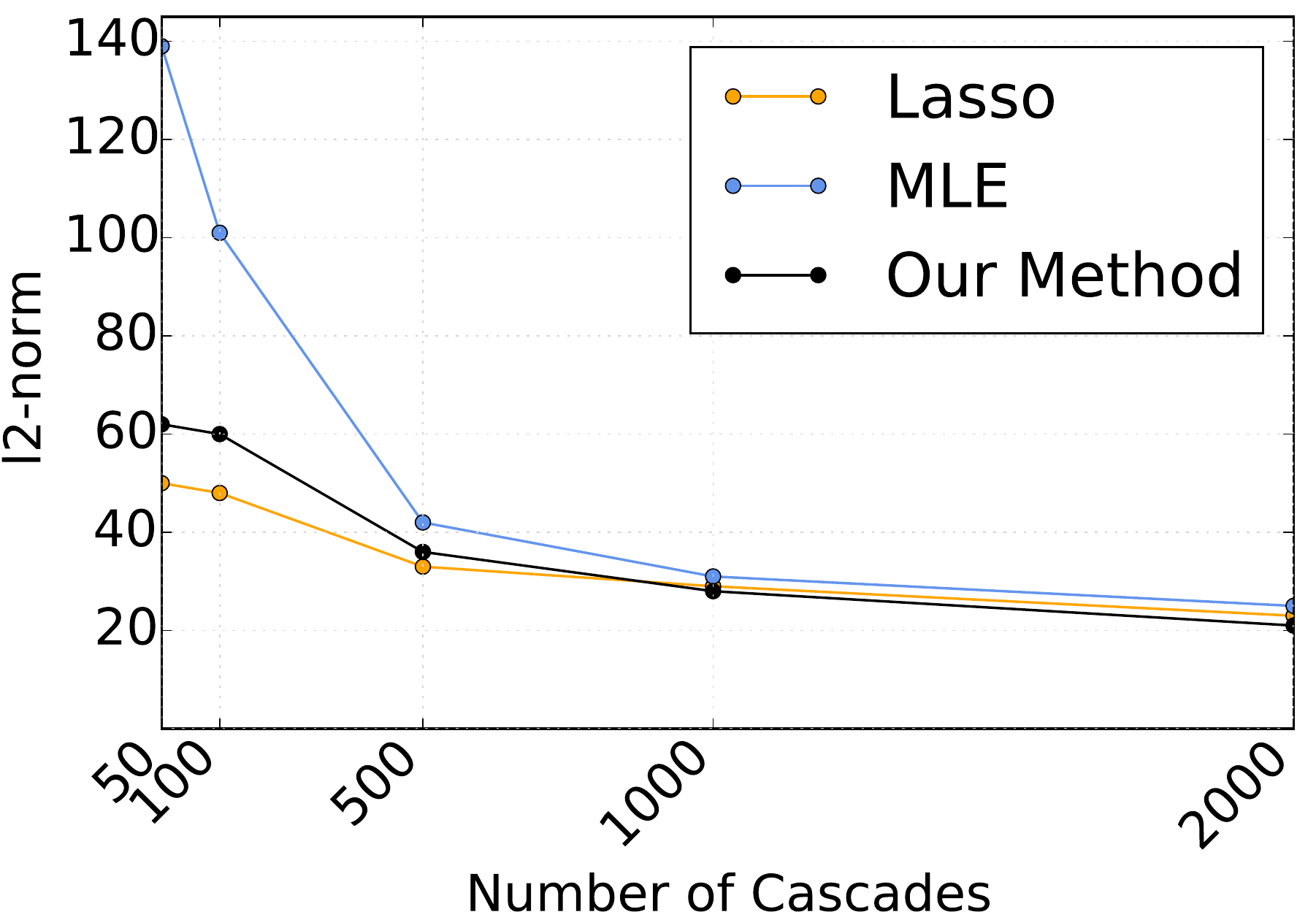}
& \hspace{-1em}\includegraphics[scale=.28]{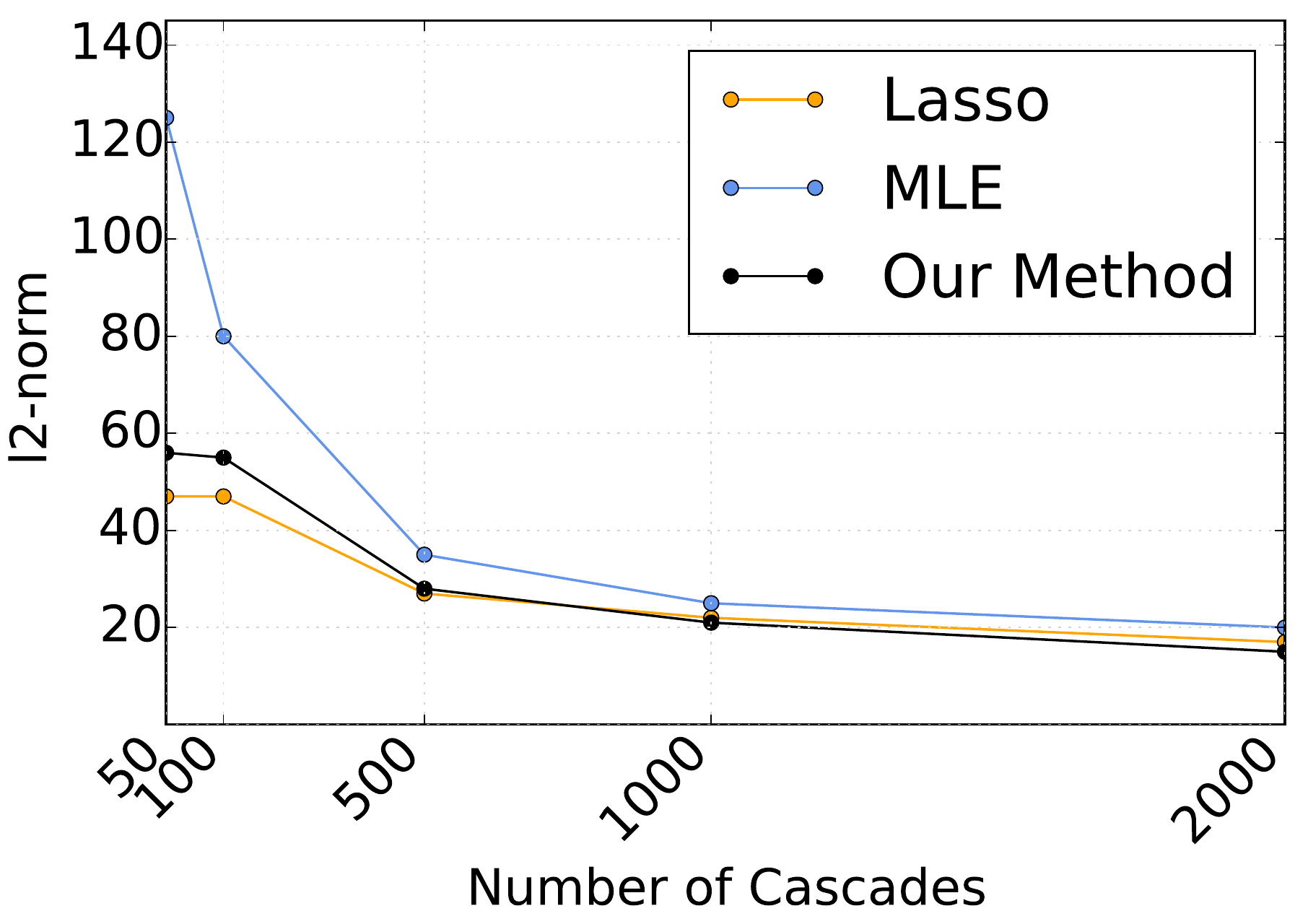}
& \hspace{-3em}\includegraphics[scale=.28]{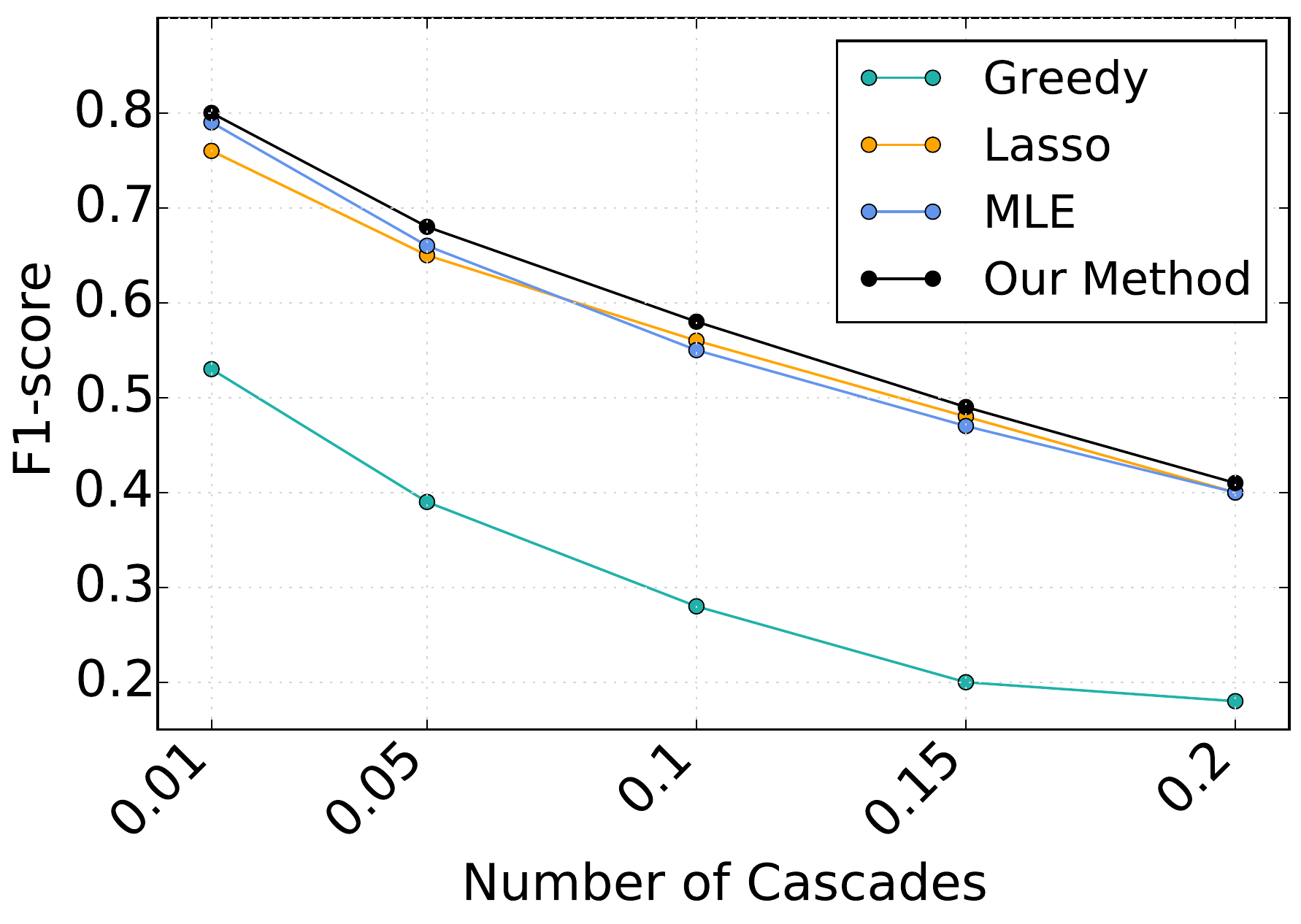} \\

(d) Sparse Kronecker ($\ell_2$-norm \emph{vs.} $n$) & (e) Non-sparse Kronecker
 ($\ell_2$-norm \emph{vs.} $n$) & (f) Watts-Strogatz (F$1$ \emph{vs.}
 $p_{\text{init}}$)
\end{tabular}
\caption{Figures (a) and (b) report the F$1$-score in $\log$ scale for
  2 graphs as a function of the number of cascades $n$: (a) Barabasi-Albert
  graph, $300$ nodes, $16200$ edges. (b) Watts-Strogatz graph, $300$ nodes,
  $4500$ edges. Figure (c) plots the Precision-Recall curve for various values
  of $\lambda$ for a Holme-Kim graph ($200$ nodes, $9772$ edges). Figures (d)
  and (e) report the $\ell_2$-norm $\|\hat \Theta - \Theta\|_2$ for a Kronecker
graph which is: (d) exactly sparse (e) non-exactly sparse, as a function of the
number of cascades $n$. Figure (f) plots the F$1$-score for the Watts-Strogatz
graph as a function of $p_{init}$.}~\label{fig:four_figs}
\vspace{-2em}
\end{figure*}

In this section, we validate empirically the results and assumptions of
Section~\ref{sec:results} for varying levels of sparsity and different
initializations of parameters ($n$, $m$, $\lambda$, $p_{\text{init}}$), where
$p_{\text{init}}$ is the initial probability of a node being a source node. We
compare our algorithm to two different state-of-the-art algorithms:
\textsc{greedy} and \textsc{mle} from~\cite{Netrapalli:2012}. As an extra
benchmark, we also introduce a new algorithm \textsc{lasso}, which approximates
our \textsc{sparse mle} algorithm.

\paragraph{Experimental setup}
We evaluate the performance of the algorithms on synthetic graphs, chosen for
their similarity to real social networks. We therefore consider a Watts-Strogatz
graph ($300$ nodes, $4500$ edges)~\cite{watts:1998}, a Barabasi-Albert graph
($300$ nodes, $16200$ edges)~\cite{barabasi:2001}, a Holme-Kim power law graph
($200$ nodes, $9772$ edges)~\cite{Holme:2002}, and the recently introduced
Kronecker graph ($256$ nodes, $10000$ edges)~\cite{Leskovec:2010}. Undirected
graphs are converted to directed graphs by doubling the edges.

For every reported data point, we sample edge weights and generate $n$ cascades
from the (IC) model for $n \in \{100, 500, 1000, 2000, 5000\}$. We compare for
each algorithm the estimated graph $\hat {\cal G}$ with ${\cal G}$. The initial
probability of a node being a source is fixed to $0.05$, i.e.\ an average of $15$
nodes source nodes per cascades for all experiments, except for
Figure~\label{fig:four_figs} (f). All edge weights are chosen uniformly in the
interval $[0.2, 0.7]$, except when testing for approximately sparse graphs (see
paragraph on robustness). Adjusting for the variance of our experiments, all
data points are reported with at most a $\pm 1$ error margin. The parameter
$\lambda$ is chosen to be of the order ${\cal O}(\sqrt{\log m / (\alpha n)})$.
We report our results as a function of the number of \emph{cascades} and not the
number of \emph{measurements}: in practice, very few cascades have depth
greater than 3.

\paragraph{Benchmarks}

We compare our \textsc{sparse mle} algorithm to 3 benchmarks: \textsc{greedy}
and \textsc{mle} from~\cite{Netrapalli:2012} and \textsc{lasso}. The
\textsc{mle} algorithm is a maximum-likelihood estimator without $\ell_1$-norm
penalization. \textsc{greedy} is an iterative algorithm. We introduced the
\textsc{lasso} algorithm in our experiments to achieve faster computation time:
$$\hat \theta_i \in \arg \min_{\theta} \sum_{t \in {\cal T}} |f(\theta_i\cdot
x^t) - x_i^{t+1}|^2 + \lambda \|\theta_i\|_1$$ \textsc{Lasso} has the merit of
being both easier and faster to optimize numerically than the other
convex-optimization based algorithms. It approximates the $\textsc{sparse mle}$
algorithm by making the assumption that the observations $x_i^{t+1}$ are of the
form: $x_i^{t+1} = f(\theta_i\cdot x^t) + \epsilon$, where $\epsilon$ is random
white noise. This is not valid in theory since $\epsilon$ \emph{depends on}
$f(\theta_i\cdot x^t)$, however the approximation is validated in practice.

We did not benchmark against other known algorithms (\textsc{netrate}
\cite{gomezbalduzzi:2011} and \textsc{first edge}~\cite{Abrahao:13}) due to the
discrete-time assumption. These algorithms also suppose a single-source model,
whereas \textsc{sparse mle}, \textsc{mle}, and \textsc{greedy} do not. Learning
the graph in the case of a multi-source cascade model is harder (see
Figure~\ref{fig:four_figs} (f)) but more realistic, since we rarely have access
to ``patient 0'' in practice.

\paragraph{Graph Estimation}
In the case of the \textsc{lasso}, \textsc{mle} and \textsc{sparse mle}
algorithms, we construct the edges of $\hat {\cal G} : \cup_{j \in V} \{(i,j) :
\Theta_{ij} > 0.1\}$, \emph{i.e} by thresholding. Finally, we report the
F1-score$=2
\text{precision}\cdot\text{recall}/(\text{precision}+\text{recall})$, which
considers \emph{(1)} the number of true edges recovered by the algorithm over
the total number of edges returned by the algorithm (\emph{precision}) and
\emph{(2)} the number of true edges recovered by the algorithm over the total
number of edges it should have recovered (\emph{recall}).  Over all experiments,
\textsc{sparse mle} achieves higher rates of precision, recall, and F1-score.
Interestingly, both \textsc{mle} and \textsc{sparse mle} perform exceptionally
well on the Watts-Strogatz graph.

\paragraph{Quantifying robustness}

The previous experiments only considered graphs with strong edges. To test the
algorithms in the approximately sparse case, we add sparse edges to the
previous graphs according to a bernoulli variable of parameter $1/3$ for every
non-edge, and drawing a weight uniformly from $[0,0.1]$. The non-sparse case is
compared to the sparse case in Figure~\ref{fig:four_figs} (d)--(e) for the
$\ell_2$ norm showing that both the \textsc{lasso}, followed by \textsc{sparse
mle} are the most robust to noise.

\section{Future Work}
\label{sec:linear_threshold}
Solving the Graph Inference problem with sparse recovery techniques opens new
venues for future work. Firstly, the sparse recovery literature has already
studied regularization patterns beyond the $\ell_1$-norm, notably the
thresholded and adaptive lasso~\cite{vandegeer:2011, Zou:2006}. Another goal
would be to obtain confidence intervals for our estimator, similarly to what
has been obtained for the Lasso in the recent series of papers
\cite{javanmard2014, zhang2014}.

Finally, the linear threshold model is a commonly studied diffusion process and
can also be cast as a \emph{generalized linear cascade} with inverse link
function $z \mapsto \mathbbm{1}_{z > 0}$:
$
    \label{eq:lt}
    X^{t+1}_j = \text{sign} \left(\inprod{\theta_j}{X^t} - t_j \right)
    $.
This model therefore falls into the 1-bit compressed sensing framework
\cite{Boufounos:2008}. Several recent papers study the theoretical
guarantees obtained for 1-bit compressed sensing with specific measurements
\cite{Gupta:2010, Plan:2014}. Whilst they obtained bounds of the order
${\cal O}(s \log \frac{m}{s}$), no current theory exists for recovering
positive bounded signals from binary measurememts. This research direction
may provide the first clues to solve the ``adaptive learning'' problem: if we
are allowed to adaptively \emph{choose} the source nodes at the beginning of
each cascade, how much can we improve the current results?

\newpage

\section*{Acknowledgments}

We would like to thank Yaron Singer, David Parkes, Jelani Nelson, Edoardo
Airoldi and Or Sheffet for helpful discussions. We are also grateful to the
anonymous reviewers for their insightful feedback and suggestions.
\bibliography{sparse}

\begin{thebibliography}{34}
\providecommand{\natexlab}[1]{#1}
\providecommand{\url}[1]{\texttt{#1}}
\expandafter\ifx\csname urlstyle\endcsname\relax
  \providecommand{\doi}[1]{doi: #1}\else
  \providecommand{\doi}{doi: \begingroup \urlstyle{rm}\Url}\fi

\bibitem[Abrahao et~al.(2013)Abrahao, Chierichetti, Kleinberg, and
  Panconesi]{Abrahao:13}
Abrahao, Bruno~D., Chierichetti, Flavio, Kleinberg, Robert, and Panconesi,
  Alessandro.
\newblock Trace complexity of network inference.
\newblock In \emph{The 19th {ACM} {SIGKDD} International Conference on
  Knowledge Discovery and Data Mining, {KDD} 2013, Chicago, IL, USA, August
  11-14, 2013}, pp.\  491--499, 2013.

\bibitem[Adar \& Adamic(2005)Adar and Adamic]{AdarA05}
Adar, Eytan and Adamic, Lada~A.
\newblock Tracking information epidemics in blogspace.
\newblock In \emph{2005 {IEEE} / {WIC} / {ACM} International Conference on Web
  Intelligence {(WI} 2005), 19-22 September 2005, Compiegne, France}, pp.\
  207--214, 2005.

\bibitem[Albert \& Barab{\'{a}}si(2001)Albert and
  Barab{\'{a}}si]{barabasi:2001}
Albert, R{\'{e}}ka and Barab{\'{a}}si, Albert{-}L{\'{a}}szl{\'{o}}.
\newblock Statistical mechanics of complex networks.
\newblock \emph{CoRR}, cond-mat/0106096, 2001.

\bibitem[Ba et~al.(2011)Ba, Indyk, Price, and Woodruff]{bipw11}
Ba, Khanh~Do, Indyk, Piotr, Price, Eric, and Woodruff, David~P.
\newblock Lower bounds for sparse recovery.
\newblock \emph{CoRR}, abs/1106.0365, 2011.

\bibitem[Bickel et~al.(2009{\natexlab{a}})Bickel, Ritov, and
  Tsybakov]{bickel2009simultaneous}
Bickel, Peter~J, Ritov, Ya'acov, and Tsybakov, Alexandre~B.
\newblock Simultaneous analysis of lasso and dantzig selector.
\newblock \emph{The Annals of Statistics}, pp.\  1705--1732,
  2009{\natexlab{a}}.

\bibitem[Bickel et~al.(2009{\natexlab{b}})Bickel, Ritov, and
  Tsybakov]{bickel:2009}
Bickel, Peter~J., Ritov, Ya’acov, and Tsybakov, Alexandre~B.
\newblock Simultaneous analysis of lasso and dantzig selector.
\newblock \emph{Ann. Statist.}, 37\penalty0 (4):\penalty0 1705--1732, 08
  2009{\natexlab{b}}.

\bibitem[Boufounos \& Baraniuk(2008)Boufounos and Baraniuk]{Boufounos:2008}
Boufounos, Petros and Baraniuk, Richard~G.
\newblock 1-bit compressive sensing.
\newblock In \emph{42nd Annual Conference on Information Sciences and Systems,
  {CISS} 2008, Princeton, NJ, USA, 19-21 March 2008}, pp.\  16--21, 2008.

\bibitem[Candes \& Tao(2006)Candes and Tao]{candes2006near}
Candes, Emmanuel~J and Tao, Terence.
\newblock Near-optimal signal recovery from random projections: Universal
  encoding strategies?
\newblock \emph{Information Theory, IEEE Transactions on}, 52\penalty0
  (12):\penalty0 5406--5425, 2006.

\bibitem[Daneshmand et~al.(2014)Daneshmand, Gomez{-}Rodriguez, Song, and
  Sch{\"{o}}lkopf]{Daneshmand:2014}
Daneshmand, Hadi, Gomez{-}Rodriguez, Manuel, Song, Le, and Sch{\"{o}}lkopf,
  Bernhard.
\newblock Estimating diffusion network structures: Recovery conditions, sample
  complexity {\&} soft-thresholding algorithm.
\newblock In \emph{Proceedings of the 31th International Conference on Machine
  Learning, {ICML} 2014, Beijing, China, 21-26 June 2014}, pp.\  793--801,
  2014.

\bibitem[Donoho(2006)]{donoho2006compressed}
Donoho, David~L.
\newblock Compressed sensing.
\newblock \emph{Information Theory, IEEE Transactions on}, 52\penalty0
  (4):\penalty0 1289--1306, 2006.

\bibitem[Du et~al.(2013)Du, Song, Woo, and Zha]{du2013uncover}
Du, Nan, Song, Le, Woo, Hyenkyun, and Zha, Hongyuan.
\newblock Uncover topic-sensitive information diffusion networks.
\newblock In \emph{Proceedings of the Sixteenth International Conference on
  Artificial Intelligence and Statistics}, pp.\  229--237, 2013.

\bibitem[Du et~al.(2014)Du, Liang, Balcan, and Song]{du2014influence}
Du, Nan, Liang, Yingyu, Balcan, Maria, and Song, Le.
\newblock Influence function learning in information diffusion networks.
\newblock In \emph{Proceedings of the 31st International Conference on Machine
  Learning (ICML-14)}, pp.\  2016--2024, 2014.

\bibitem[Gomez~Rodriguez et~al.(2010)Gomez~Rodriguez, Leskovec, and
  Krause]{GomezRodriguez:2010}
Gomez~Rodriguez, Manuel, Leskovec, Jure, and Krause, Andreas.
\newblock Inferring networks of diffusion and influence.
\newblock In \emph{Proceedings of the 16th ACM SIGKDD International Conference
  on Knowledge Discovery and Data Mining}, KDD '10, pp.\  1019--1028, New York,
  NY, USA, 2010. ACM.
\newblock ISBN 978-1-4503-0055-1.

\bibitem[Gomez{-}Rodriguez et~al.(2011)Gomez{-}Rodriguez, Balduzzi, and
  Sch{\"{o}}lkopf]{gomezbalduzzi:2011}
Gomez{-}Rodriguez, Manuel, Balduzzi, David, and Sch{\"{o}}lkopf, Bernhard.
\newblock Uncovering the temporal dynamics of diffusion networks.
\newblock \emph{CoRR}, abs/1105.0697, 2011.

\bibitem[Gupta et~al.(2010)Gupta, Nowak, and Recht]{Gupta:2010}
Gupta, Ankit, Nowak, Robert, and Recht, Benjamin.
\newblock Sample complexity for 1-bit compressed sensing and sparse
  classification.
\newblock In \emph{{IEEE} International Symposium on Information Theory, {ISIT}
  2010, June 13-18, 2010, Austin, Texas, USA, Proceedings}, pp.\  1553--1557,
  2010.

\bibitem[Holme \& Kim(2002)Holme and Kim]{Holme:2002}
Holme, Petter and Kim, Beom~Jun.
\newblock Growing scale-free networks with tunable clustering.
\newblock \emph{Physical review E}, 65:\penalty0 026--107, 2002.

\bibitem[Javanmard \& Montanari(2014)Javanmard and Montanari]{javanmard2014}
Javanmard, Adel and Montanari, Andrea.
\newblock Confidence intervals and hypothesis testing for high-dimensional
  regression.
\newblock \emph{The Journal of Machine Learning Research}, 15\penalty0
  (1):\penalty0 2869--2909, 2014.

\bibitem[Kempe et~al.(2003)Kempe, Kleinberg, and Tardos]{Kempe:03}
Kempe, David, Kleinberg, Jon~M., and Tardos, {\'{E}}va.
\newblock Maximizing the spread of influence through a social network.
\newblock In \emph{Proceedings of the Ninth {ACM} {SIGKDD} International
  Conference on Knowledge Discovery and Data Mining, Washington, DC, USA,
  August 24 - 27, 2003}, pp.\  137--146, 2003.

\bibitem[Leskovec et~al.(2007)Leskovec, McGlohon, Faloutsos, Glance, and
  Hurst]{Leskovec07}
Leskovec, Jure, McGlohon, Mary, Faloutsos, Christos, Glance, Natalie~S., and
  Hurst, Matthew.
\newblock Patterns of cascading behavior in large blog graphs.
\newblock In \emph{Proceedings of the Seventh {SIAM} International Conference
  on Data Mining, April 26-28, 2007, Minneapolis, Minnesota, {USA}}, pp.\
  551--556, 2007.

\bibitem[Leskovec et~al.(2010)Leskovec, Chakrabarti, Kleinberg, Faloutsos, and
  Ghahramani]{Leskovec:2010}
Leskovec, Jure, Chakrabarti, Deepayan, Kleinberg, Jon~M., Faloutsos, Christos,
  and Ghahramani, Zoubin.
\newblock Kronecker graphs: An approach to modeling networks.
\newblock \emph{Journal of Machine Learning Research}, 11:\penalty0 985--1042,
  2010.

\bibitem[Liben-Nowell \& Kleinberg(2008)Liben-Nowell and Kleinberg]{Nowell08}
Liben-Nowell, David and Kleinberg, Jon.
\newblock {Tracing information flow on a global scale using Internet
  chain-letter data}.
\newblock \emph{Proceedings of the National Academy of Sciences}, 105\penalty0
  (12):\penalty0 4633--4638, 2008.

\bibitem[Negahban et~al.(2012)Negahban, Ravikumar, Wrainwright, and
  Yu]{Negahban:2009}
Negahban, Sahand~N., Ravikumar, Pradeep, Wrainwright, Martin~J., and Yu, Bin.
\newblock A unified framework for high-dimensional analysis of m-estimators
  with decomposable regularizers.
\newblock \emph{Statistical Science}, 27\penalty0 (4):\penalty0 538--557,
  December 2012.

\bibitem[Netrapalli \& Sanghavi(2012)Netrapalli and Sanghavi]{Netrapalli:2012}
Netrapalli, Praneeth and Sanghavi, Sujay.
\newblock Learning the graph of epidemic cascades.
\newblock \emph{SIGMETRICS Perform. Eval. Rev.}, 40\penalty0 (1), June 2012.
\newblock ISSN 0163-5999.

\bibitem[Plan \& Vershynin(2014)Plan and Vershynin]{Plan:2014}
Plan, Yaniv and Vershynin, Roman.
\newblock Dimension reduction by random hyperplane tessellations.
\newblock \emph{Discrete {\&} Computational Geometry}, 51\penalty0
  (2):\penalty0 438--461, 2014.

\bibitem[Price \& Woodruff(2011)Price and Woodruff]{pw11}
Price, Eric and Woodruff, David~P.
\newblock {(1} + eps)-approximate sparse recovery.
\newblock In Ostrovsky, Rafail (ed.), \emph{{IEEE} 52nd Annual Symposium on
  Foundations of Computer Science, {FOCS} 2011, Palm Springs, CA, USA, October
  22-25, 2011}, pp.\  295--304. {IEEE} Computer Society, 2011.
\newblock ISBN 978-1-4577-1843-4.

\bibitem[Price \& Woodruff(2012)Price and Woodruff]{pw12}
Price, Eric and Woodruff, David~P.
\newblock Applications of the shannon-hartley theorem to data streams and
  sparse recovery.
\newblock In \emph{Proceedings of the 2012 {IEEE} International Symposium on
  Information Theory, {ISIT} 2012, Cambridge, MA, USA, July 1-6, 2012}, pp.\
  2446--2450. {IEEE}, 2012.
\newblock ISBN 978-1-4673-2580-6.

\bibitem[Raskutti et~al.(2010)Raskutti, Wainwright, and Yu]{raskutti:10}
Raskutti, Garvesh, Wainwright, Martin~J., and Yu, Bin.
\newblock Restricted eigenvalue properties for correlated gaussian designs.
\newblock \emph{Journal of Machine Learning Research}, 11:\penalty0 2241--2259,
  2010.

\bibitem[Rudelson \& Zhou(2013)Rudelson and Zhou]{rudelson:13}
Rudelson, Mark and Zhou, Shuheng.
\newblock Reconstruction from anisotropic random measurements.
\newblock \emph{{IEEE} Transactions on Information Theory}, 59\penalty0
  (6):\penalty0 3434--3447, 2013.

\bibitem[van~de Geer et~al.(2011)van~de Geer, Bühlmann, and
  Zhou]{vandegeer:2011}
van~de Geer, Sara, Bühlmann, Peter, and Zhou, Shuheng.
\newblock The adaptive and the thresholded lasso for potentially misspecified
  models (and a lower bound for the lasso).
\newblock \emph{Electron. J. Statist.}, 5:\penalty0 688--749, 2011.

\bibitem[van~de Geer \& B{\"u}hlmann(2009)van~de Geer and
  B{\"u}hlmann]{vandegeer:2009}
van~de Geer, Sara~A. and B{\"u}hlmann, Peter.
\newblock On the conditions used to prove oracle results for the lasso.
\newblock \emph{Electron. J. Statist.}, 3:\penalty0 1360--1392, 2009.

\bibitem[Watts \& Strogatz(1998)Watts and Strogatz]{watts:1998}
Watts, Duncan~J. and Strogatz, Steven~H.
\newblock Collective dynamics of `small-world' networks.
\newblock \emph{Nature}, 393\penalty0 (6684):\penalty0 440--442, 1998.

\bibitem[Zhang \& Zhang(2014)Zhang and Zhang]{zhang2014}
Zhang, Cun-Hui and Zhang, Stephanie~S.
\newblock Confidence intervals for low dimensional parameters in high
  dimensional linear models.
\newblock \emph{Journal of the Royal Statistical Society: Series B (Statistical
  Methodology)}, 76\penalty0 (1):\penalty0 217--242, 2014.

\bibitem[Zhao \& Yu(2006)Zhao and Yu]{Zhao:2006}
Zhao, Peng and Yu, Bin.
\newblock On model selection consistency of lasso.
\newblock \emph{J. Mach. Learn. Res.}, 7:\penalty0 2541--2563, December 2006.
\newblock ISSN 1532-4435.

\bibitem[Zou(2006)]{Zou:2006}
Zou, Hui.
\newblock The adaptive lasso and its oracle properties.
\newblock \emph{Journal of the American Statistical Association}, 101\penalty0
  (476):\penalty0 1418--1429, 2006.

\end{thebibliography}
\bibliographystyle{icml2015}

\newpage
\section{Appendix}
In this appendix, we provide the missing proofs of Section~\ref{sec:results}
and Section~\ref{sec:lowerbound}. We also show additional experiments on the
running time of our recovery algorithm which could not fit in the main part of
the paper.

\subsection{Proofs of Section~\ref{sec:results}}

\begin{proof}[Proof of Lemma~\ref{lem:transform}]
Using the inequality $\forall x>0, \; \log x \geq 1 - \frac{1}{x}$, we have
$|\log (\frac{1}{1 - p}) - \log (\frac{1}{1-p'})| \geq \max(1 - \frac{1-p}{1-p'},
1 - \frac{1-p'}{1-p}) \geq \max( p-p', p'-p)$.
\end{proof}

\begin{proof}[Proof of Lemma~\ref{lem:ub}]
The gradient of $\mathcal{L}$ is given by:
\begin{multline*}
    \nabla \mathcal{L}(\theta^*) =
    \frac{1}{|\mathcal{T}|}\sum_{t\in \mathcal{T}}x^t\bigg[
        x_i^{t+1}\frac{f'}{f}(\inprod{\theta^*}{x^t})\\
    - (1-x_i^{t+1})\frac{f'}{1-f}(\inprod{\theta^*}{x^t})\bigg]
\end{multline*}

Let $\partial_j \mathcal{L}(\theta)$ be the $j$-th coordinate of
$\nabla\mathcal{L}(\theta^*)$.  Writing
$\partial_j\mathcal{L}(\theta^*)
= \frac{1}{|\mathcal{T}|}\sum_{t\in\mathcal{T}} Y_t$ and since
$\E[x_i^{t+1}|x^t]= f(\inprod{\theta^*}{x^t})$, we have that $\E[Y_{t+1}|Y_t]
= 0$. Hence $Z_t = \sum_{k=1}^t Y_k$ is a martingale.

Using assumption (LF), we have almost surely $|Z_{t+1}-Z_t|\leq
\frac{1}{\alpha}$ and we can apply Azuma's inequality to $Z_t$:
\begin{displaymath}
    \P\big[|Z_{\mathcal{T}}|\geq \lambda\big]\leq
    2\exp\left(\frac{-\lambda^2\alpha}{2n}\right)
\end{displaymath}

Applying a union bound to have the previous inequality hold for all coordinates
of $\nabla\mathcal{L}(\theta)$ implies:
\begin{align*}
    \P\big[\|\nabla\mathcal{L}(\theta^*)\|_{\infty}\geq \lambda \big]
    &\leq 2m\exp\left(\frac{-\lambda^2n\alpha}{2}\right)
\end{align*}
Choosing $\lambda\defeq 2\sqrt{\frac{\log m}{\alpha n^{1-\delta}}}$ concludes
the proof.
\end{proof}

\begin{proof}[Proof of Corollary~\ref{cor:variable_selection}]
By choosing $\delta = 0$, if $ n > \frac{9s\log m}{\alpha\gamma^2\epsilon^2}$,
then $\|\hat \theta-\theta^*\|_2 < \epsilon < \eta$ with probability
$1-\frac{1}{m}$. If $\theta_i^* = 0$ and $\hat \theta > \eta$, then $\|\hat
\theta - \theta^*\|_2 \geq |\hat \theta_i-\theta_i^*| > \eta$, which is a
contradiction. Therefore we get no false positives. If $\theta_i^* > \eta +
\epsilon$, then $|\hat{\theta}_i- \theta_i^*| < \epsilon \implies \theta_j >
\eta$ and we get all strong parents.
\end{proof}

\paragraph{(RE) with high probability} We now prove Proposition~\ref{prop:fi}.
The proof mostly relies on showing that the Hessian of likelihood function
$\mathcal{L}$ is sufficiently well concentrated around its expectation.

\begin{proof}Writing $H\defeq \nabla^2\mathcal{L}(\theta^*)$, if
    $ \forall\Delta\in C(S),\;
        \|\E[H] - H]\|_\infty\leq \lambda $
    and $\E[H]$ verifies the $(S,\gamma)$-(RE)
    condition then:
    \begin{equation}
        \label{eq:foo}
    \forall \Delta\in C(S),\;
    \Delta H\Delta \geq
    \Delta \E[H]\Delta(1-32s\lambda/\gamma)
    \end{equation}
    Indeed, $
    |\Delta(H-E[H])\Delta| \leq 2\lambda \|\Delta\|_1^2\leq
    2\lambda(4\sqrt{s}\|\Delta_s\|_2)^2
    $.
    Writing
    $\partial^2_{i,j}\mathcal{L}(\theta^*)=\frac{1}{|\mathcal{T}|}\sum_{t\in
    T}Y_t$ and using $(LF)$ and $(LF2)$ we have $\big|Y_t - \E[Y_t]\big|\leq
    \frac{3}{\alpha}$.
    Applying Azuma's inequality as in the proof of Lemma~\ref{lem:ub}, this
    implies:
    \begin{displaymath}
        \P\big[\|\E[H]-H\|_{\infty}\geq\lambda\big] \leq
        2\exp\left(-\frac{n\alpha\lambda^2}{3} + 2\log m\right)
    \end{displaymath}
    Thus, if we take $\lambda=\sqrt{\frac{9log m}{\alpha
    n^{1-\delta}}}$, $\|E[H]-H\|_{\infty}\leq\lambda$ w.p at least
    $1-e^{-n^{\delta}\log m}$. When $n^{1-\delta}\geq
    \frac{1}{28\gamma\alpha}s^2\log m$, \eqref{eq:foo} implies
    $
    \forall \Delta\in C(S),\;
    \Delta H\Delta \geq \frac{1}{2} \Delta \E[H]\Delta,
    $ w.p.~at least $1-e^{-n^{\delta}\log m}$ and the conclusion of
    Proposition~\ref{prop:fi} follows.
\end{proof}

\subsection{Proof of Theorem~\ref{thm:lb}}

Let us consider an algorithm $\mathcal{A}$ which verifies the recovery
guarantee of Theorem~\ref{thm:lb}: there exists a probability distribution over
measurements such that for all vectors $\theta^*$, \eqref{eq:lb} holds w.p.
$\delta$. This implies by the probabilistic method that for all distribution
$D$ over vectors $\theta$, there exists an $n\times m$ measurement matrix $X_D$
with such that \eqref{eq:lb} holds w.p. $\delta$ ($\theta$ is now the random
variable).

Consider the following distribution $D$: choose $S$ uniformly at random from a
``well-chosen'' set of $s$-sparse supports $\mathcal{F}$ and $t$ uniformly at
random from
$X \defeq\big\{t\in\{-1,0,1\}^m\,|\, \mathrm{supp}(t)\in\mathcal{F}\big\}$.
Define $\theta = t + w$ where
$w\sim\mathcal{N}(0, \alpha\frac{s}{m}I_m)$ and $\alpha = \Omega(\frac{1}{C})$.

Consider the following communication game between Alice and Bob: \emph{(1)}
Alice sends $y\in\reals^m$ drawn from a Bernouilli distribution of parameter
$f(X_D\theta)$ to Bob. \emph{(2)} Bob uses $\mathcal{A}$ to recover
$\hat{\theta}$ from $y$.  It can be shown that at the end of the game Bob now
has a quantity of information $\Omega(s\log \frac{m}{s})$ about $S$. By the
Shannon-Hartley theorem, this information is also upper-bounded by $\O(n\log
C)$. These two bounds together imply the theorem.

\subsection{Running Time Analysis}

\begin{figure}
  \centering
  \includegraphics[scale=.4]{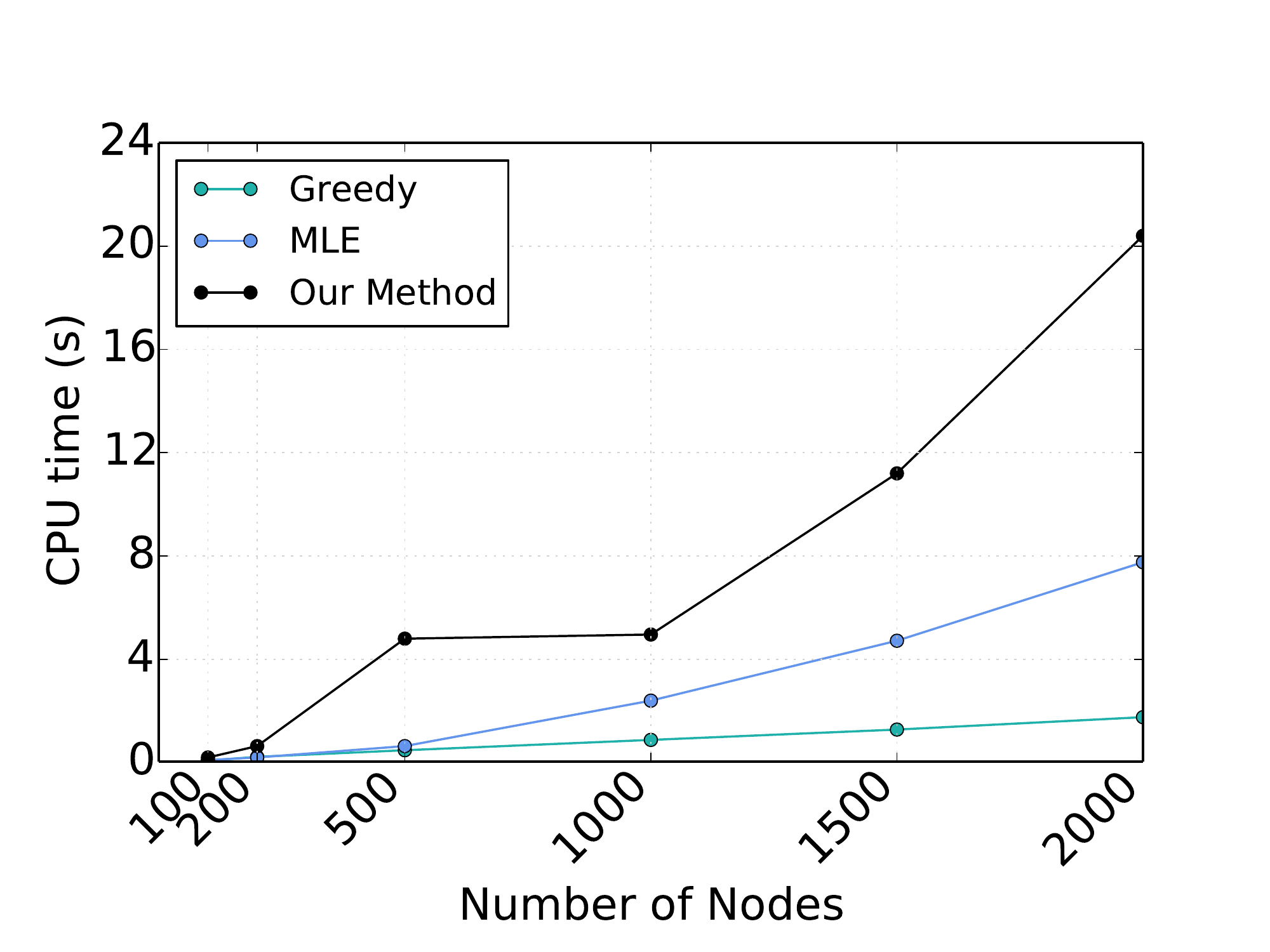}
  \vspace{-1em}
  \caption{Running time analysis for estimating the parents of a \emph{single
    node} on a Barabasi-Albert graph as a function of the number of nodes in
    the graph. The parameter $k$ (number of nodes each new node is attached to)
    was set to $30$.  $p_{\text{init}}$ is chosen equal to $.15$, and
    the edge weights are chosen uniformly at random in $[.2,.7]$. The
  penalization parameter $\lambda$ is chosen equal to $.1$.}
    \label{fig:running_time_n_nodes}
    \vspace{-1em}
\end{figure}

\begin{figure}
  \centering
  \includegraphics[scale=.4]{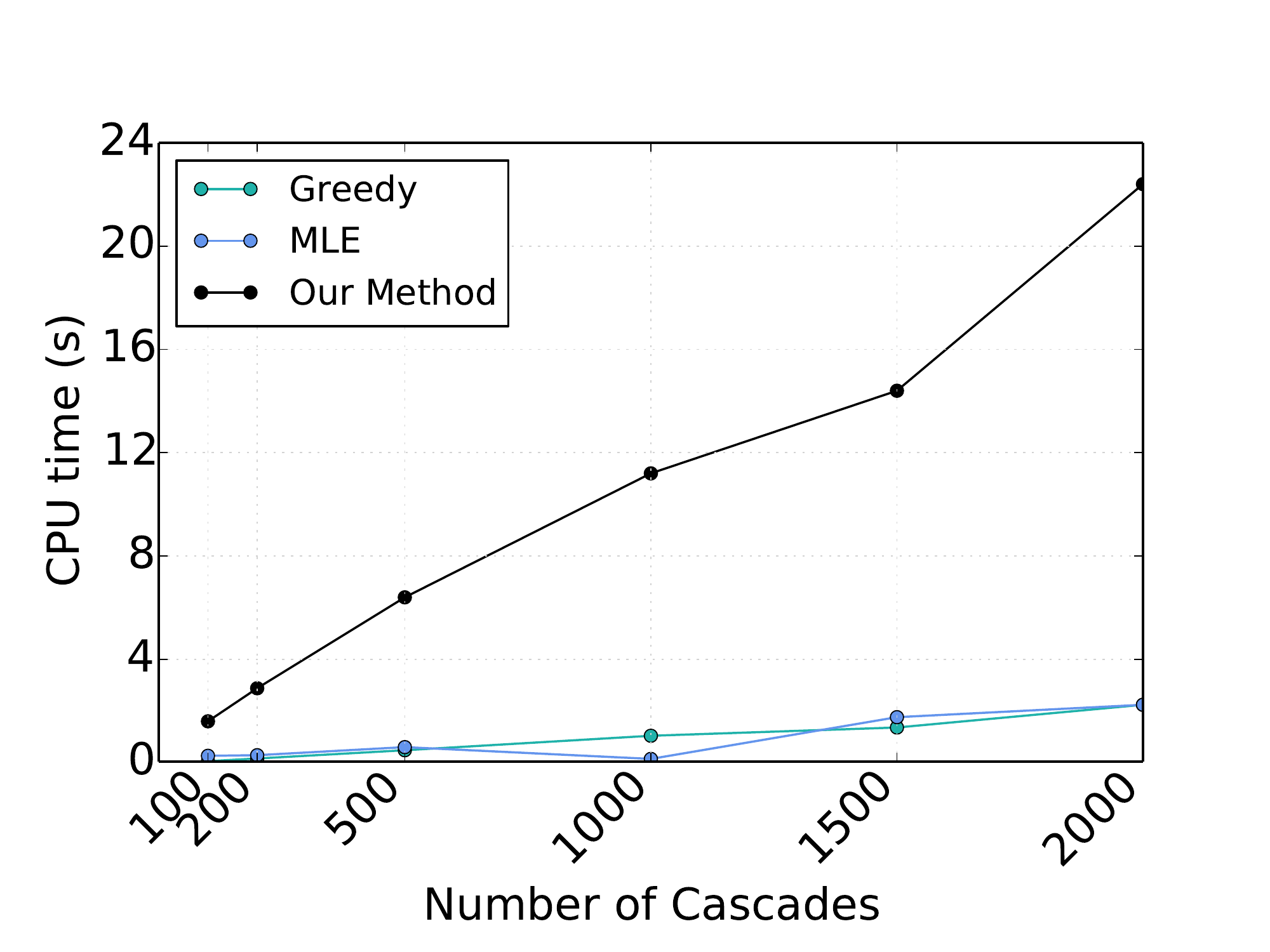}
  \caption{Running time analysis for estimating the parents of a \emph{single
    node} on a Barabasi-Albert graph as a function of the number of total
    observed cascades. The parameters defining the graph were set as in
Figure~\ref{fig:running_time_n_nodes}.}
    \label{fig:running_time_n_cascades}
    \vspace{-1em}
\end{figure}

We include here a running time analysis of our algorithm. In
Figure~\ref{fig:running_time_n_nodes}, we compared our algorithm to the
benchmark algorithms for increasing values of the number of nodes. In
Figure~\ref{fig:running_time_n_cascades}, we compared our algorithm to the
benchmarks for a fixed graph but for increasing number of observed cascades.

In both Figures, unsurprisingly, the simple greedy algorithm is the fastest.
Even though both the MLE algorithm and the algorithm we introduced are based on
convex optimization, the MLE algorithm is faster. This is due to the overhead
caused by the $\ell_1$-regularisation in~\eqref{eq:pre-mle}.

The dependency of the running time on the number of cascades increases is
linear, as expected. The slope is largest for our algorithm, which is again
caused by the overhead induced by the $\ell_1$-regularization.









\end{document}